\documentclass[a4paper,onecolumn,11pt, accepted=2025-02-10]{quantumarticle}
\pdfoutput=1
\usepackage[utf8]{inputenc}
\usepackage[english]{babel}
\usepackage[T1]{fontenc}
\usepackage{amsmath}
\usepackage{hyperref}
\usepackage{amssymb,amsmath,amsthm}
\usepackage{tikz}
\usepackage{braket}
\usepackage{lipsum}
\usepackage{dcolumn}        
\usepackage{bm}
\usepackage[numbers,sort&compress]{natbib}
\usepackage{tikz}
\usepackage{slashed}
\usepackage{mathtools}
\usepackage{algorithm}
\usepackage{enumerate}
\usepackage{algpseudocode}
\newtheorem{theorem}{Theorem}[section]
\newtheorem{corollary}[theorem]{Corollary}
\newtheorem{definition}[theorem]{Definition}
\newtheorem{lemma}[theorem]{Lemma}
\newtheorem{proposition}[theorem]{Proposition}
\newtheorem{remark}[theorem]{Remark}
\newtheorem{example}[theorem]{Example}
\definecolor{myblue}{RGB}{0,0,255}
\definecolor{mygreen}{RGB}{0,128,0}
\def\bZ {\mathbb{Z}}

\DeclarePairedDelimiter\floor{\lfloor}{\rfloor}

\begin{document}

\title{Synthesis and Arithmetic of Single Qutrit Circuits}

\author{Amolak Ratan Kalra}
\affiliation{Institute for Quantum Computing, University of Waterloo, Waterloo, Ontario, Canada}
\affiliation{David R. Cheriton School of Computer Science, University of Waterloo, Waterloo, Ontario, Canada}
\affiliation{Perimeter Institute for Theoretical Physics, Waterloo, Ontario, Canada}
\author{Michele Mosca}
\affiliation{Institute for Quantum Computing, University of Waterloo, Waterloo, Ontario, Canada}
\affiliation{Dept. of Combinatorics $\&$ Optimization, University of Waterloo, Waterloo, Ontario, Canada}
\affiliation{Perimeter Institute for Theoretical Physics, Waterloo, Ontario, Canada}
\author{Dinesh Valluri}
\affiliation{Institute for Quantum Computing, University of Waterloo, Waterloo, Ontario, Canada}
\affiliation{Dept. of Combinatorics $\&$ Optimization, University of Waterloo, Waterloo, Ontario, Canada}

\maketitle

\begin{abstract}
In this paper we study single qutrit circuits consisting of words over the Clifford$+\mathcal{D}$ cyclotomic gate set, where $\mathcal{D}=\text{diag}(\pm\xi^{a},\pm\xi^{b},\pm\xi^{c})$, $\xi$ is a primitive $9$-th root of unity and $a,b,c$ are integers. We characterize classes of qutrit unit vectors $z$ with entries in $\mathbb{Z}[\xi, \frac{1}{\chi}]$ based on the possibility of reducing their smallest denominator exponent (sde) with respect to $\chi := 1 - \xi,$ by acting an appropriate gate in Clifford$+\mathcal{D}$. We do this by studying the notion of `derivatives mod $3$' of an arbitrary element of $\mathbb{Z}[\xi]$ and using it to study the smallest denominator exponent of $HDz$ where $H$ is the qutrit Hadamard gate and $D \in \mathcal{D}.$ 
In addition, we reduce the problem of finding all unit vectors of a given sde to that of finding integral solutions of a positive definite quadratic form along with some additional constraints. As a consequence we prove that the Clifford$+\mathcal{D}$ gates naturally arise as gates with sde $0$ and $3$ in the group $U(3,\mathbb{Z}[\xi, \frac{1}{\chi}])$ of $3 \times 3$ unitaries with entries in $\mathbb{Z}[\xi, \frac{1}{\chi}]$. We illustrate the general applicability of these methods to obtain an exact synthesis algorithm for Clifford$+R$ and recover the previous exact synthesis algorithm in \cite{kmm}. The framework developed to formulate qutrit gate synthesis for Clifford$+\mathcal{D}$ extends to qudits of arbitrary prime power. 
\end{abstract}

\section{Introduction}
\label{intro}
Implementing a target unitary $U$ using a finite sequence of gates over a universal gate set $G$ is a fundamental problem in quantum computing.
\noindent This problem arises naturally in the context of quantum compilation, where the goal is to map an abstract instruction set to a physical architecture. The abstract instruction set usually specifies a quantum algorithm or protocol and the physical architecture or hardware usually refers to a quantum circuit.
\noindent A circuit synthesis algorithm takes as input a target unitary $U$, a universal gate set $G$ and a precision parameter $\epsilon$ and outputs a sequence of gates from $G$ that realizes/implements the unitary $U$ up-to $\epsilon$ precision efficiently.

\medskip

\noindent The general synthesis problem benefits greatly from the subroutine of exact synthesis which is: given a finite universal gate set $G$ and an element $g$ in the group generated by $G$ find a sequence $g_{1}...g_{k}$ in $G$ such that $g=g_{1}...g_{k}$. In the qubit case Kliuchnikov, Maslov and Mosca (KMM) \cite{kmm} solved the exact synthesis problem. In particular they showed that single qubit unitaries $U$ which are exactly implementable (decomposable) over the Clifford$+T$ gate set are equivalent to the $2 \times 2$ unitaries over the ring $\bZ[\frac{1}{\sqrt{2}},i]$. Great progress was made in designing efficient synthesis algorithms by using this exact synthesis result as a sub-routine, see \cite{kmm, ross2014optimal,yardapprox}. 

\medskip

\noindent Motivated by the results in these papers we aim to formulate and address the analogous problem of exact synthesis for qutrits over various families of universal gate sets.
The motivation for studying quantum circuits based on qutrits and more generally qudits comes from recent experimental progress made on implementing and controlling higher dimensional quantum systems in the lab \cite{low2023control, subramanian2023efficient,PhysRevLett.114.240401,PhysRevA.67.062313,C5CS00933B,Ringbauer_2022}.
From a theoretical perspective, magic state distillation protocols based on qudits have been shown to outperform their qubit counterparts, see \cite{PhysRevA.92.022312,PhysRevX.2.041021,prakash2024low}. Qudits have also been used to make the connection between quantum speed-up and contextuality more transparent \cite{Howard_2014, PhysRevA.101.010303}. There are several other experimental and theoretical advantages for considering qudits over qubits with reference to circuit complexity and ease of experimental implementations, these are discussed in some detail in \cite{qditrev}. The question of circuit synthesis for qudits has partially been addressed previously using normal forms, see \cite{prakash2018normal,GLAUDELL201954, prakash2021normal}.

\medskip

\noindent A direct generalization of the main result of \cite{kmm} fails for qutrit circuit synthesis. In other words, the Clifford$+T$ gates are insufficient to generate $U(3,R_{9,\chi}),$ the analogous group of $3 \times 3$ unitary matrices whose entries are in $R_{9,\chi} := \mathbb{Z}[\xi,\frac{1}{\chi}],$ where $\xi = e^{\frac{2\pi i}{9}},$ a primitive $9$-th root of unity and $\chi = 1-\xi.$ 
In this paper we rescue the situation for qutrits by proposing two gate sets that satisfy an algebraic characterization analogous to \cite{kmm}. To this end, we study the qutrit analogue of \cite{kmm} over the Clifford$+\mathcal{D}$ gate set where $\mathcal{D}=\text{diag}(\pm \xi^{a},\pm \xi^{b}, \pm \xi^{c}),$ for integers $a,b,c.$ The choice of these gates is not arbitrary, as we show in Section \ref{unit vectors} that the $\mathcal{D}$ gates and the $H$ gate arise as, in a sense, the only gates whose smallest denominator exponent (sde) is $0$ and $3$ respectively.  This sets us up with the question of whether Clifford$+\mathcal{D}$ adequately generates $U(3,R_{9,\chi}),$ i.e., given a matrix $M \in U(3,R_{9, \chi}),$ can it be generated by gates from Clifford$+\mathcal{D}$. If so, does there exist a reasonably efficient algorithm to do so? In this paper we make progress towards answering this question by a study of the smallest denominator exponents of unit vectors with entries in $R_{9,\chi}$.

\medskip

\noindent The following property holds for $G =$ Clifford$+T$ in $ U(2,R_{8,\chi})$ and $G = $ Clifford$+R$ in $ U(3,R_{3,\chi})$:

\begin{center}
    \textit{
P: Given any unit vector $z\in R_{n,\chi}^3$ with $sde\geq c $ there exists an element $g \in G$ such that $\text{sde}(gz,\chi )< \text{sde}(z,\chi)$}, where $c$ is a constant 
less than or equal to 3 
 only depending on $n.$
\end{center}

\noindent Existence of property P for a given $G \subset U(p,R_{\chi})$ allows us to find a sequence of letters  $g_1, \ldots, g_k$ such that the $\text{sde}(g_1\ldots g_k z) = c-1$. This reduces the problem of exact synthesis to a finite computation.

However, P need not hold for every $p^{l}$. For instance, whenever $p=3,l=2$ we prove in Theorem \ref{maintheorem} that there exists a $\Delta\in \bZ_{3}$ associated to the unit vector $z$ in such a way that there exists a $g \in $ Clifford$+\mathcal{D}$ such that  $sde(gz) < sde(z)$ iff $\Delta\neq -1$. Therefore, we identify an easily computable obstruction $\Delta$ for property P. While $\Delta$ is an obstruction, our method in fact helps us recover $g$ explicitly when $\Delta \neq -1.$    

\medskip

\noindent The general strategy of employing property P, if it is true, makes sense for a general prime power $n = p^{l}$ and in fact most of the tools we develop, such as the notion of derivatives mod $p$, do extend to such generality. However, in the case where $p = 2$ or $3,$ by Lemma \ref{equalsde} the sdes of each entry of a given unitary in $U_{p}(R_{\chi})$ are  equal, a property which fails to be true when $p \geq 5$. This failure makes the sde profile for $p \geq 5$ more complicated to study. A key idea to prove our results is the reduction of property P for a particular instance of $z$ to the existence of $\mathbb{Z}_{p}$- rational solutions of a system of polynomial equations over $\mathbb{Z}_{p}.$ For $p \leq 3$ the equations are simple enough (at most quadratic) to be solved completely.

\medskip

\noindent A brief summary of the contents are as follows. In Section \ref{prelims}, we review the background material on rings and universal gate sets required for the rest of the paper. In Section \ref{derivative}, we formulate the notion of `derivatives mod p’ of an element $f(\zeta)$ of $\mathbb{Z}[\zeta]$, where $\zeta$ is a primitive $p^{l}$-th root of unity. These are simply the elements $\frac{f^{(k)}(1)}{k!} \pmod p $ for $0 \leq k \leq \phi(p^l)$. We prove that this notion is well defined, i.e., if $f(x),g(x) \in \mathbb{Z}[x]$ such that $f(\zeta) = g(\zeta)$, then $f^{k}(1)/k! = g^{k}(1)/k! \pmod p$. This essentially follows from the fact that $f(x) = g(x) \pmod {\Phi_{p^l}(x)}$, which implies $f^{(k)}(1)/k! = g^{(k)}(1)/k!$ modulo an integral linear combination of $\Phi_{p^l}^{(k)}(1)/k!$ which divisible by $p$ for all $0 \leq k \leq \phi(p^l)$ by Lemma \ref{phider}. Here $\Phi_{n}(x)$ is the $n$-th cyclotomic polynomial, a monic irreducible polynomial over $\mathbb{Z}.$ An intuitive way to see these derivatives is as Taylor coefficients of $f$ at `$\zeta = 1$’. This can be made precise by considering the quotient map $\mathbb{Z}[\zeta] \xrightarrow{} \mathbb{Z}[\zeta]/(p)$ and looking at the image of $f(\zeta)$ in terms of $\chi = 1-\zeta$. We prove a criterion which is repeatedly used: for $0 \leq k \leq \phi(p^l),$ $\chi^{k}$ divides $f(\zeta)$ iff $f(1) = \frac{f’(1)}{1!} =\ldots= \frac{f^{k-1}(1)}{(k-1)!} = 0 \pmod p$. This is in fact easy to see using the Taylor series picture presented above.

\medskip

\noindent In Section \ref{unit vectors}, we find an algorithm to obtain all the unit vectors with a given sde. The unit vector condition over $R_{9,\chi} := \mathbb{Z}[\xi]_{\chi}$ with respect to the hermitian metric, namely $\sum |z_{i}|^2 = 1,$ can be reduced to solutions of three non-homogeneous integral quadratic forms in 18 variables. One of the quadratic forms  is positive definite therefore the set of its integral solutions is finite. This allows us to find all the unit vectors of a given sde by finding the integral solutions of the positive definite quadratic form and checking if the solution obtained satisfies the other two equations. This way of obtaining all the unit vectors of a particular sde is practical for small sde. Using this we can characterize the unitary matrices of sde $= 0$ as precisely the ones generated by the Pauli gates + $\mathcal{D}$ gates, while the gates with sde $= 3$ are, up to permutation of their columns, the $H$ gate, possibly multiplied by two gates from $\mathcal{D}$ on the left and right respectively. Therefore, we obtain the Clifford$+\mathcal{D}$ gates as natural objects in the arithmetic of the group $U(3,R_{9,\chi})$. 

\medskip

\noindent In Section \ref{mainresults}, we study the difference $\text{sde}(HDR^{\epsilon}z) - \text{sde}(z)$, where $D$ varies over the cyclotomic gate set $\mathcal{D}$ in the context of single qubit and single  qutrit gates. This question of dropping sde of $z$ can be tied to the question of divisibility of $z$ by powers of $\chi := 1- \zeta_{p^l}$. We reformulate this condition of divisibility by powers of $\chi $ in terms of derivatives vanishing mod $p$, using Theorem \ref{derivativecriterion}. This condition gives rise to a system of polynomial equations as explained in the proofs of Theorems \ref{cliffrthm} and \ref{maintheorem}. This is done for $p^l = 2^3$ to recover the main theorem of \cite{kmm}. For, $p^l = 3$ we give an exact synthesis algorithm for Clifford$+R$ and prove a result (as an easy consequence) analogous to the one in \cite{kmm}, i.e., Clifford$+R$ = $U(3,R_{\chi}).$ This algorithm is closely connected to the approximate synthesis algorithm for Clifford$+R$ explored in \cite{vadymanyons}.
   For $p^l = 3^2$ we obtain a criterion  $\text{sde}(HDR^{\epsilon}z) - \text{sde}(z) = -1 \text{ iff } \Delta \neq -1, \text{ where } \Delta \in \mathbb{Z}_3$ is an element associated to the unit vector $z$.

\medskip

\noindent \textbf{Historical Note}: After this paper appeared, the exact synthesis problem posed in this paper was completely solved for the group $U_3(\mathbb{Z}[\xi]_{\chi})$, where $\xi$ is a primitive $9$-th root of unity and $\chi = 1 - \xi$ by Shai Evra and Ori Parzanchevski in \cite{evra2024arithmeticity}. In particular, they showed that Clifford$+\mathcal{D}$ is a set of generators for $U_3(\mathbb{Z}[\xi]_{\chi})$ using the theory of buildings. This extends the qubit synthesis results of \cite{kmm}  to qutrits. In particular, the current work gives the frame work and proposes the candidate set of gates, namely Clifford$+\mathcal{D}$ gates in place of the Clifford$+T$ gates for qubits. The question of multi-qutrit exact synthesis was later addressed later by two independent works \cite{Glaudell_2024, kalra2024multi}.
\section{Preliminaries}
\label{prelims}
In this introductory section, we recall basic definitions of various qubit and qutrit gate sets, and the various cyclotomic rings and their localizations and the notions of smallest denominator exponent and greatest dividing exponent used throughout the paper.
\begin{definition}
The single qubit Clifford$+T$ gate set is generated by the following matrices:
\[
H=\frac{1}{\sqrt{2}}\begin{bmatrix}
1 &  1\\
1 & -1
\end{bmatrix}~~~~
T=\begin{bmatrix}
1 & 0\\
0 & \zeta_{8}
\end{bmatrix}
\]
where $\zeta_{8}=e^\frac{2\pi i}{8}$ is the primitive 8th root of unity.
\end{definition}
\noindent We will also be interested in qutrit analogue of the single qubit Clifford$+T$ group which is defined as follows:
\begin{definition}
The single qutrit Clifford$+T$ group is generated by the following matrices:
\[
H=-\frac{i}{\sqrt{3}}
\begin{bmatrix}
1 & 1 & 1\\
1 &  \omega & \omega^{2}\\
1 & \omega^{2} & \omega
\end{bmatrix}~~~~
S=\begin{bmatrix}
1 & 0 & 0\\
0 & \omega & 0\\
0 & 0 & 1
\end{bmatrix}~~~~
T=\begin{bmatrix}
\xi & 0 & 0\\
0 & 1 & 0\\
0 & 0 & \xi^{-1}
\end{bmatrix}
\]
where $\xi=e^\frac{{2\pi i}}{9}$ which is the primitive $9$th root of unity.
\end{definition}
\noindent For qutrits, another universal gate set referred to as the  Clifford$+R$ gate set \cite{metaplectic} is defined as follows:
\begin{definition}
The single qutrit Clifford$+R$ gate set is generated by:
\end{definition}
\[
H=\frac{1}{\sqrt{-3}}
\begin{bmatrix}
1 & 1 & 1\\
1 &  \omega & \omega^{2}\\
1 & \omega^{2} & \omega
\end{bmatrix}~~~
S=\begin{bmatrix}
1 & 0 & 0\\
0 & \omega & 0\\
0 & 0 & 1
\end{bmatrix}~~~~
R=\begin{bmatrix}
1 & 0 & 0\\
0 & 1 & 0\\
0 & 0 & -1
\end{bmatrix}
\]
The $R$ gate is sometimes referred to as the metaplectic gate.
As shown in \cite{metaplectic}, the  multi-qutrit Clifford$+R$ gate set is a subset of multi-qutrit Clifford$+T$ gate set.
We define the Clifford$+\mathcal{D}$ gate set which can be thought of as the single qutrit analogue of the Clifford Cyclotomic gate set discussed in \cite{Forest_2015}:
\begin{definition}
The single qutrit Clifford$+\mathcal{D}$ gate set is generated by the following gates:
\[
H=\frac{1}{\sqrt{-3}}
\begin{bmatrix}
1 & 1 & 1\\
1 &  \omega & \omega^{2}\\
1 & \omega^{2} & \omega
\end{bmatrix}~~~
D_{[a,b,c]}=\begin{bmatrix}
\xi^{a} & 0 & 0\\
0 & \xi^{b} & 0\\
0 & 0 & \xi^{c}
\end{bmatrix}~~~
R_{[a,b,c]}=
\begin{bmatrix}
(-1)^a & 0 & 0\\
0 & (-1)^b & 0\\
0 & 0 & (-1)^c
\end{bmatrix}
\]

where $\xi=e^\frac{2\pi i}{9}.$ Here we denote all gates of the form $\text{diag}(\pm \xi^a, \pm \xi^b, \pm \xi^c)$ as $\mathcal{D}.$ Often we drop the subscript $[a,b,c]$ from $D_{[a,b,c]}$ and simply write $D$. We denote $R := R_{[0,0,1]}$ as it is frequently used. 
\end{definition}
\noindent Note that a $D$ gate is a diagonal operator in the Clifford Hierarchy \cite{Cui_2017}.
More generally, one can define the qudit version of the Clifford$+T$ gate set for odd prime dimensions as follows:
\begin{definition}
The single qudit Clifford$+T$ group is generated by the following:
\[
H=\frac{1}{\sqrt{p}}\sum_{i=0}^{p-1}\sum_{j=0}^{p-1}{\zeta_p}^{ij}\ket{i}\bra{j}~~~S=\sum_{j=0}^{p-1}{\zeta_p}^{j(j+1)2^{-1}}\ket{j}\bra{j}~~~
T= \sum_{j=0}^{p-1}\zeta_p^{j^{3}6^{-1}}\ket{j}\bra{j}
\]
\end{definition}
\noindent Note that for qudits of odd prime dimension $p>3$, the $T$ gate requires only the powers of $\zeta_{p}$ and no higher roots of unity as explained in \cite{PhysRevA.86.022316}.
It is well known that these single qutrit gate sets are universal for quantum computation when supplemented with an appropriately chosen two qutrit entangling gate.

\noindent We now recall the definition of certain rings which are used later in the paper and are closely connected to exact synthesis. Note that the ring $R$ ``localized" at $\chi$ is denoted as $R_{\chi}$:

\begin{definition}
Let $p$ be a prime, $l\geq 1$ an integer, denote $n=p^l,$ and $R_n := \mathbb{Z}[\zeta_{n}]$ be the ring of cyclotomic integers, i.e., $\zeta_{n}$ is a primitive $n$-th root of unity and $R_{n,\chi} := \{\frac{a}{\chi^f} : a \in R_n\},$ where $\chi=1-\zeta_{p^{n}}$ i.e., it is the ring obtained by allowing powers of $\chi$ in the denominators. We denote by $U(p,R_{n,\chi})$ the group of $p \times p$ unitary matrices whose entries are in $R_{n, \chi}.$ We will suppress the subscript $n$ for both $R_n$ and $R_{n,\chi}$ when it is clear from the context. We will also write $\zeta$ in place of $\zeta_n,$ when $n$ is clear from the context.
\end{definition}

\noindent It is easy to see that $U(p,R_{\chi})$ is a group since the inverse of a unitary is its conjugate transpose and that $\overline{\chi} = 1 - \zeta^{-1} = -\zeta^{-1} \chi$. 

\begin{remark}~
\begin{enumerate}
\item When $n = 2^3$, the ring $R_{\chi}$ is precisely $\mathbb{Z}[e^{\frac{2\pi i}{8}},1/\sqrt{2}] = \mathbb{Z}[i,1/\sqrt{2}]$ the ring in \cite{kmm}. When $n = 3$ the ring is the one considered in \cite{vadymanyons}.
\item For arbitrary $n = p^l$, the gates Clifford$+\mathcal{D}$ (which includes the T gate) are contained in $U(p,R_\chi).$
\end{enumerate}
\end{remark}
\noindent More explicitly, we can describe the rings $R_{n}$ and $R_{n,\chi}$ in the cases where $n = 3, 8$ and $9$ as follows:
\begin{itemize}
\item $R_{8} =\{\sum_{i=0}^{i=3}a_{i}\zeta_{8}^{i}~|~a_{i}\in \bZ,~\zeta_{8}=e^\frac{{2\pi i}}{8}\}$
\item $R_{8,\chi} =\{\frac{a}{\chi^{f}}~|~a\in R_{8},~f\in \bZ
_{\geq 0}\}$ 
\item$R_{9}=\bZ[\xi]:=\{\sum_{i=0}^{5}a_{i}\xi^{i}~|~ a_{i}\in \bZ,~\xi=e^\frac{{2\pi i}}{9}\}$
\item $R_{9,\chi}=\{\frac{a}{\chi^{f}}~|~a\in R_{9},~f\in \bZ_{\geq 0}\}$ 
\item $R_{3}=\{a_{0}+a_{1}\omega~|~a_{i}\in \bZ,~\omega=e^\frac{{2\pi i}}{3}\}$
\item $R_{3,\chi}=\{\frac{a}{\chi^{f}}~|~a\in R_{3},~f\in \bZ_{\geq 0}\}$ 
\end{itemize}
In general, for qudits we will have the ring $R_{\chi}=\{\frac{a}{\chi^{f}}~|~a\in R,~f\geq 0\}$ where $R=\bZ[\zeta]$ and $\zeta$ is the primitive $n$th root of unity for some $n=p^{l}$.

\noindent We will now define two key notions which are in fact the main tool to study exact synthesis:
\begin{definition}
The smallest denominator exponent (sde) of $z \in R_{\chi}$ with respect to $\chi$ is defined as the smallest non-negative integer $f$ such that $\chi^{f}z\in R$.
\end{definition}
\begin{definition}
The greatest dividing exponent (gde) of an element $f(\zeta) \in \mathbb{Z}[\zeta]$ with respect to $\chi$ is defined as the largest integer $k$ such that $\chi^k$ divides $f(\zeta)$.
\end{definition}

\noindent The notions of sde and gde are closely related. If $u = \frac{a}{\chi^f},$ for $a \in \mathbb{Z}[\xi],$ then we have $$sde(u) = \max\{0, f - gde(a)\}.$$

\begin{example}
Let $g = \frac{f(\xi)}{3},$ where $f(\xi) = 1 + \xi + \xi^2 \in \mathbb{Z}[\xi].$ We see that $f(\xi) = 1 + (1-\chi) + (1-\chi)^2 = 3 - \chi - 2\chi + \chi^2 = u\chi^6 -u\chi^7 + \chi^2 = \chi^2(1 + u\chi^4 - u\chi^5).$ Therefore, $\text{gde}(f(\xi),\chi) = 2$ and $\text{sde}(g,\chi) = 4.$

\end{example}

\begin{example}
    Let $n=3$, i.e., $\omega = e^{\frac{2\pi i}{3}}$. Since $\omega$ satisfies $\omega^2 + \omega + 1 = 0,$ an arbitrary element of $\mathbb{Z}[\omega]$ can be written as $\alpha = a+b \omega,$ where $a,b \in \mathbb{Z}$. Note that $3 = \chi^2 u,$ for some unit $u \in \mathbb{Z}[\omega].$ Therefore,  $\chi$ divides $\alpha$ iff $a+b = 0 \pmod{3}.$

\noindent If $a = 1 \pmod{3}$ and $b = -1 \pmod{3}$ then $\chi^{2}$ does not divide $\alpha$. Indeed, if $\chi^2$ divides $\alpha$ then both $a$ and $b$ are divisible by $3$.  Therefore, $gde(\alpha, \chi) = 1.$
\end{example}

\noindent We briefly discuss a tool to compute the sde of an arbitrary element of $R_{n,\chi}$ for any $n = p^l$, to be fully described in the next section. Consider the  map $$P : \mathbb{Z}[\zeta] \xrightarrow[]{f(\zeta) \mapsto f(1)} \mathbb{Z}/p$$

\noindent The map is a well-defined ring homomorphism with $\ker(P) = \chi$. Indeed, if $f(\zeta) = g(\zeta),$ where $f,g \in \mathbb{Z}[x],$ then $f(x) \equiv g(x) \pmod{\Phi_{n}(x)}.$ It is well known that $\Phi_{n}(1) = 0 \pmod{p}$ whenever $n$ is a power of $p.$ Therefore, $f(1) \equiv g(1) \pmod{p}.$  Note that $f(1) = 0 \pmod{p}$ iff $\chi$ divides $f(\zeta).$ This essentially follows from the remainder theorem, i.e., $f(\zeta) = f(1) + \chi g(\zeta),$ for some $g(\zeta) \in \mathbb{Z}[\zeta]$ and the fact that $p = u\chi^{\phi(n)} = u \chi^{p^{l-1}(p-1)},$ for some unit $u$ in $\mathbb{Z}[\zeta].$ Note that the map $P$ already captures the divisibility of an element in $\mathbb{Z}[\zeta]$ by $\chi.$ In the next section, we introduce higher analogues of $P$ which capture the divisibility of an element in $\mathbb{Z}[\zeta]$ by $\chi^k.$

\section{Derivatives mod $p$}
\label{derivative}
In this section, we shall make sense of differentiation of elements of the ring of cyclotomic integers $\mathbb{Z}[\zeta]$, modulo a prime $p$, where $\zeta$ is a primitive $n=p^l$-th root of unity, for $l \geq 1$. 
The morphism $\mathbb{Z}[x] \xrightarrow{} \mathbb{Z}[\zeta]$ defined by $f(x) \mapsto f(\zeta)$ is a ring homomorphism whose kernel is generated by $\Phi_{n}(x),$ the $n$-th cyclotomic polynomial. In other words, $f(\zeta) = g(\zeta)$ iff $f(x) \equiv g(x) \pmod{\Phi_{n}(x)}.$ The polynomial $\Phi_{n}(x)$ is monic and irreducible in the ring $\mathbb{Z}[x]$ and has degree $\phi(n) = p^{l-1}(p-1),$ the Euler totient of $n$. More explicitly, $$\Phi_n(x) = \prod_{k:\text{ }gcd(k,n)=1} (x-\zeta^{k})$$

\noindent It is well known that: $$x^n - 1 = \prod_{d|n}\Phi_d(x).$$

\noindent By applying the above identity for $n = p \text{ and } p^2,$  we have
$ x^{p} - 1 = \Phi_1(x) \Phi_p(x),$
$x^{p^2} - 1 = \Phi_1(x)\Phi_p(x)\Phi_{p^2}(x)$
where $\Phi_{1}(x) = x-1.$ These imply $\Phi_{p}(x) = x^{p-1} + \ldots + 1$ and $\Phi_{p^2}(x) = (x^{p})^{p-1} + (x^{p})^{p-2} \ldots + x^{p} + 1$ respectively. In particular, when $p=3,$ $\Phi_3(x) = x^2 + x + 1$ and $\Phi_9(x) = x^{6} + x^3 + 1.$ Similarly $\Phi_{8}(x) = x^4 + 1.$ More generally, we have $$\Phi_{p^{l}}(x) = \frac{x^{p^l} - 1}{x^{p^{l-1}}-1} = \sum_{j=0}^{p-1} x^{jp^{l-1}}$$

\begin{remark}
    By substituting $x=1$ in the above identity we obtain $$p = \prod_{k:\text{ } \gcd(k,n)=1} (1-\zeta^{k}) = (1-\zeta)^{\phi(n)} \prod_{k:\text{ } \gcd(k,n)=1} \frac{1-\zeta^{k}}{1-\zeta} $$
    For integers $k$ such that $gcd(k,n) = 1,$ $\frac{1-\zeta^{k}}{1-\zeta}$ is in fact a unit. Therefore, we may write $p = (1-\zeta)^{\phi(n)} u$ for some unit $u$ in $\mathbb{Z}[\zeta]$. 
    \end{remark}

\noindent The following lemma follows easily from the above identity. Although, we only need this lemma for $n = 3,8,9$ and $n = p$ for prime $p \geq 5.$
\begin{lemma}   \label{phider}
For $0 \leq k < \phi(p^{l}) = p^{l-1}(p-1),$ 
$$\frac{\Phi_{p^l}^{(k)}(1)}{k!}  \equiv 0 \pmod{p}$$
\end{lemma}
\begin{proof}
If $y = x^{p^{l-1}}-1$ then the above identity can be written as $$\Phi_{p^{l}}(x) = \frac{(y+1)^{p} - 1}{y} = y^{p-1} + {p\choose p-1}y^{p-2} + \ldots + {p \choose 1}$$
\noindent Now the lemma follows from the fact that $p$ divides ${p \choose j}$ for $1 \leq j \leq p-1.$
\end{proof}

\begin{example}
   We may verify the above lemma for $n=9$ by explicit computation. Using $\Phi_{9}(x) = x^6 + x^3 + 1,$ we may compute 
$\frac{\Phi_{9}^{(1)}(1)}{1!} = 6x^5 + 3x^{2}|_{x=1} = 9,$
$\frac{\Phi_{9}^{(2)}(1)}{2!} = 30x^4 + 6x|_{x=1} = 36,$
$\frac{\Phi_{9}^{(3)}(1)}{3!} = \frac{1}{3!}120x^3 + 6|_{x=1} = 21,$
$\frac{\Phi_{9}^{(4)}(1)}{4!} = \frac{1}{4!}360x^2|_{x=1} = 360/24 = 15,$  $\frac{\Phi_{9}^{(5)}(1)}{5!} = \frac{1}{5!}720x|_{x=1} = 6,$
and $\frac{\Phi_{9}^{(6)}(1)}{6!} = 720/6! = 1.$ Observe that $\frac{\Phi_{9}^{(k)}(1)}{k!} \equiv 0 \pmod{3}$ for $0 \leq k \leq 5.$ 
\end{example}

\begin{proposition}
    Let $f(x), g(x) \in \mathbb{Z}[x]$ such that $f(\zeta) = g(\zeta)$, where $\zeta$ is an $n$-th root of unity, $n = p^{l},l \geq 1$. For every $0 \leq k < \phi(n) = p^{l-1}(p-1)$ 
    $$\frac{f^{(k)}(1)}{k!} = \frac{g^{(k)}(1)}{k!} \pmod{p}$$
\end{proposition}

\begin{proof}
    As we noted before, $f(\zeta) = g(\zeta)$ if and only if $f(x) = g(x) \pmod{\Phi_{n}(x)},$ i,e., $f(x) = g(x) + h(x)\Phi_{n}(x)$ for some $h(x) \in \Phi_{n}(x).$ By differentiating on both sides $k$ times and dividing by $k!$ we get 
    $$\frac{f^{(k)}(x)}{k!} = \frac{g^{(k)}(x)}{k!} + \sum_{r=0}^{k}  \frac{h^{(r)}(x)}{r!}\frac{\Phi_n^{(k-r)}(x)}{(k-r)!}$$
    By substituting $x=1$, the proposition follows from Lemma \ref{phider}.  
\end{proof}

\begin{definition}
    Let $\zeta = \zeta_{p^l}$ be a primitive $p^l$-th root of unity and $f(\zeta) \in \mathbb{Z}[\zeta]$. For $0 \leq k \leq \phi(p^l)$ we define the $k$-th derivative modulo $p$ as $\frac{f^{(k)}(1)}{k!}  \pmod{p}.$ 
\end{definition}

\noindent The above proposition makes the notion of `$k$-th derivative mod $p$' well-defined.

\noindent Let $q: \mathbb{Z}[\zeta] \xrightarrow{} \mathbb{Z}[\zeta]/(p)$ be the usual quotient map mod $p$. For $f(x)\in \mathbb{Z}[x],$  suppose $f(x) = \sum a_{i} (1-x)^{i}.$  We have $\frac{f^{(k)}(1)}{k!} = a_{k}$ for all $k.$ In particular, when $x = \zeta$ and $\chi := 1-\zeta,$ $f(\zeta) = \sum a_{i} \chi^{i}.$ Therefore, the `derivatives mod $p$' are precisely the Taylor coefficients of $f(x)$ modulo $p$ at $x=1$.
\begin{remark} \label{reduction}
    Note that since $\Phi_{n}(x),$ the minimal polynomial of $\zeta$, is monic of degree $\phi(n),$ we may always assume that the degree of $f(x)$ is less than $\phi(n)$. 
    Therefore, $\frac{f^{(k)}(1)}{k!} = 0 \pmod{p}$ for $k > \phi(n).$
\end{remark}

\noindent The following theorem relates these derivatives with the gde of elements of $\mathbb{Z}[\zeta]$, therefore giving a straight forward and efficient algorithm to compute the gde of any arbitrary element of $\mathbb{Z}[\zeta].$
\begin{theorem} \label{derivativecriterion}
    Let $n=p^{l},$ $l \geq 1$ and $1\leq k \leq \phi(n)$ for a prime $p$. Let $f(x)$ be an arbitrary polynomial with integer coefficients then 
    \begin{center}
        $\chi^{k}$ divides $f(\zeta)$ iff $f(1) = \frac{f^{(1)}(1)}{1!} = \ldots = \frac{f^{(k-1)}(1)}{(k-1)!} = 0 \mod p$
    \end{center}
\end{theorem}
\begin{proof}
     We may write $$f(\zeta) = f(1-\chi) = \sum_{i=0}^{d} a_i {(-\chi)}^{i} $$
We may assume $d < \phi(n)$ by Remark \ref{reduction}.  Suppose $\chi^k$ divides $f(\zeta)$ for $k \geq 1$ then it is immediate that $\chi$ divides $a_0$. Since $a_0$ is an integer, the norm $N(\chi) = p$ divides $a_0.$ However, $p = u\chi^{\phi(n)}$ for some unit $u$ in $\mathbb{Z}[\zeta].$ In particular, this implies $\chi$ divides $a_1$ which again implies $N(\chi) = p$ divides $a_1$. Carrying on iteratively, we may obtain $a_0 = \ldots = a_{k-1} = 0 \pmod{p}$. Conversely, if $a_0 = \ldots = a_{k-1} = 0 \pmod{3}$ then $\chi^k$ obviously divides $f(\xi)$ by the above Taylor expansion for $f$. This concludes the proof of the theorem. 
\end{proof}

\section{Unit vectors and unitary matrices with entries in $R_{\chi} = \{\frac{a}{\chi^{f}} : a \in \mathbb{Z}[\zeta]\}$}
\label{unit vectors}

In this section, we characterize qutrit unit vectors of a given sde $f$ as integer vectors $a := (a_0,\ldots, a_5), b := (b_0, \ldots, b_5),c := (c_0, \ldots, c_5) \in \mathbb{Z}^{6}$ satisfying certain integral quadratic forms. This yields us an algorithm to give a list of all unit vectors of sde $f$ whose complexity is equivalent to the complexity of finding $a,b,c \in \mathbb{Z}^6$ such that $$q(a) + q(b) + q(c) = N$$ 

\noindent where $N = O(3^{{\frac{f}{3}}})$ and $q(a) := \frac{1}{4} ((2a_{0} - a_{3})^2 + 3a_{3}^2 + (2a_{1} - a_{4})^2 + 3a_{4}^2 + (2a_{2} - a_{5})^2 + 3a_{5}^2)$.

\medskip

\noindent Let $R = \mathbb{Z}[\xi], \xi = e^{\frac{2 \pi i}{9}},\chi = 1-\xi$ and $R_{\chi} := \{\frac{a}{\chi^{f}} : a \in R\}.$ In this section, we study qutrit unit vectors
with entries in $R_{\chi},$ i.e., $z = (z_{1},z_{2},z_{3}) \in R_{\chi}^{3}$ such that $|z_1|^2 + |z_2|^2 + |z_3|^2 = 1$. Note that we may write $z = \frac{1}{\chi^{f}}(w_1,w_2,w_3),$ where $\chi \nmid w_{i},$ for some some $i$ and an integer $f \geq 0$. The integer $f$ is the smallest denominator exponent (sde) of $z$.

\medskip

\noindent Recall that the map $P : \mathbb{Z}[\xi] \xrightarrow[]{} \mathbb{Z}_3,$ defined by $P(f(\xi)) = f(1) \pmod{3}$ is well defined. The unit vector condition can be rewritten as 
\begin{align} \label{unitvector}
    |w_{1}|^2 + |w_{2}|^2 + |w_{3}|^2 = |\chi|^{2f}
\end{align}
If $f \geq 1,$ by applying the map $P$ to this equation gives $\sum_{i=1}^3 w_{i}(1)^2 = 0.$ This implies  that $w_{i}(1) = \pm 1,$ for all $i$, in other words $\chi \nmid w_{i}$ for all $i.$

\medskip

\noindent Let $\tau := \xi + \xi^{-1}.$ Note that the subring $\mathbb{Z}[\tau] \subset \mathbb{Z}[\xi]$ is precisely the set of real elements of $\mathbb{Z}[\xi]$. The minimal polynomial of $\tau$ is $x^{3} - 3x + 1.$ Therefore $1, \tau, \tau^{2}$ form a basis of $\mathbb{Z}[\tau]$ over $\mathbb{Z}.$ Since $|w_{i}|^{2}$ are real numbers, we may express them as an integral linear combination of $1, \tau, \tau^2.$ This allows us to reduce
Equation \ref{unitvector} to a system of equations of integral quadratic forms by writing $|w_{i}|^2$ on the LHS and $|\chi|^2$ on the RHS with respect to the basis  $1, \tau, \tau^{2}$. First, an element $f(\xi)$ of $\mathbb{Z}[\xi]$ can be written uniquely as $f(\xi) = \sum_{i=0}^{5} a_{i} \xi^{i}$, where $a_{i} \in \mathbb{Z}.$ We have 
\begin{align}
 |f(\xi)|^2 & = \sum_{0 \leq i,j \leq 5} a_{i} a_{j} \xi^{i} \xi^{-j}\\&   = \sum_{i=0}^{5}  a_{i}^{2} + \sum_{1 \leq i < j \leq 5}    a_{i}a_{j} (\xi^{i-j} + \xi^{j-i}) \\ &= 
 \sum_{i=0}^{5} a_{i}^2 +  \sum_{k=1}^{5} (\sum_{i < j, |i-j| = k} a_{i}a_{j}) (\xi^{k} + \xi^{-k})
\end{align}

\noindent We denote  $\alpha_{k} = \sum _{i \leq j, |i-j| = k} a_{i} a_{j}.$ Therefore, $|f(\xi)|^2 = \alpha_{0} + \sum_{k=1}^5 \alpha_k \tau_{k}.$

\noindent Recall that $\xi^{9} = 1$ and the minimal polynomial of $\xi$ is the cyclotomic polynomial $\Phi_{9}(x) = x^6 + x^3 + 1.$ Denote $\tau_{k} = \xi^{k} + \xi^{-k},$ for $0 \leq k \leq 5,$ and $\tau := \tau_{1}$. It is straightforward to write $\tau_{k}$ in terms of $\tau$ as follows 
\begin{align}
    \tau_{2} &= \tau^2 - 2 \\
    \tau_{3} &= \xi^{3} + \xi^{-3} = -1 \\
    \tau_{4} &= -\tau^2 - \tau + 2 \\
    \tau_{5} &= \tau_{4} = -\tau^2 - \tau + 2
\end{align}

\noindent The expression for $\tau_4$ is derived as $ \tau^{4} = \xi^{4} + 4 \xi^{3} \xi^{-1} + 6 \xi^{2} \xi^{-2} + 4 \xi^{1} \xi^{-3} + \xi^{-4} = \ \tau_{4} + 4\tau_2 + 6 = 3\tau^2 - \tau \implies \tau_{4} = 3\tau^{2} - \tau -4\tau_{2} - 6 = 3\tau^2 - \tau-4(\tau^2 - 2) - 6  = -\tau^2 - \tau + 2.$ Therefore, we find 
\begin{align}
    |f(\xi)|^2 &= \alpha_0 + \alpha_1 \tau + \alpha_2 (\tau^2 - 2) + \alpha_3 (-1) + \alpha_4 (-\tau^2 - \tau + 2) + \alpha_5 (-\tau^2 -\tau + 2) \\
    &= \underbrace{(\alpha_0 -2\alpha_2 - \alpha_3+ 2\alpha_4 + 2\alpha_5)}_{q_{0}} + \underbrace{(\alpha_1 - \alpha_4 - \alpha_5)}_{q_{1}}\tau + \underbrace{(\alpha_2 - \alpha_4 - \alpha_5)}_{q_2}\tau^2
\end{align}

\noindent We may view $q_{i} = q_{i}(a_0, \ldots, a_5)$ (also written as $q_{i}(a)$) as quadratic forms in the six variables $a = (a_0, \ldots, a_5)$. In particular, if $w_{1}(\xi) = \sum_{j=0}^{5} a_{j} \xi^j,$ $w_{2}(\xi) = \sum_{j=0}^5 b_j \xi^{j}$ and $w_3 = \sum_{j=0}^{5} c_{j} \xi^j$ then $|w_{1}|^2 = q_{0}(a) + q_{1}(a) \tau + q_{2}(a)\tau^2, |w_{2}|^2 = q_{0}(b) + q_{1}(b) \tau + q_{2}(b)\tau^2$ and $|w_{3}|^2 = q_{0}(c) + q_{1}(c) \tau + q_{2}(c)\tau^2.$ Therefore, the unit vector condition implies
$$(q_{0}(a) + q_{0}(b) + q_{0}(c)) + (q_{1}(a) + q_{1}(b) + q_{1}(c)) \tau + (q_{2}(a) + q_{2}(b) + q_{2}(c)) \tau^2 = |\chi|^{2f} = A + B\tau + C\tau^2,$$
for some integers $A,B$ and $C$. By the linear independence of $\{1, \tau, \tau^2\}$ in $\mathbb{Z}[\tau]$ over $\mathbb{Z},$ we have 
\begin{align*}
    q_{0}(a) + q_{0}(b) + q_{0}(c) = A\\
    q_{1}(a) + q_{1}(b) + q_{1}(c) = B\\
    q_{2}(a) + q_{2}(b) + q_{2}(c) = C
\end{align*}

\noindent Note that given $f$, the integers $A,B$ and $C$ are easily computable by writing $|\chi|^2 = (1-\xi)(1-\xi^{-1}) = 2 - \tau,$  and therefore using the binomial theorem to expand $|\chi|^{2f} = (2-\tau)^f.$ Since $\tau$ satisfies its minimal polynomial $\tau^{3} - 3\tau + 1 = 0,$ we can write any $\tau^{k}$ for $k \geq 3$ as an integral linear combination of $1, \tau$ and $\tau^2.$ In the following examples, we fully describe $A,B$ and $C$ using a minor simplification.

\begin{example}~
    \begin{enumerate}
        \item For $f = 0,$ $(A,B,C) = (1,0,0).$ 
        \item For $f = 1,$ $(A,B,C) = (2,-1,0)$.
        \item For $f=2,$ $|\chi|^{2f} = (2-\tau)^2 = 4 - 4\tau + \tau^2,$ therefore $(A,B,C)  = (4,-4,1).$
        \item For $f=3,$ $\chi^{2f} = (2-\tau)^3 = 8 -12\tau + 6\tau^2 - \tau^3 = 9 - 15\tau + 6\tau^2.$ Therefore, $(A,B,C) = (9,-15,6)$. However, in this case we can simplify the above quadratic form equations by observing that $3 = \chi^{6}u$ for some unit $u$ in $\mathbb{Z}[\xi].$ Indeed, simply multiply $|u|^2$ on both sides of Equation \ref{unitvector}, which gives us $3$ on the RHS. Therefore, we may assume $(A,B,C) = (3,0,0)$. 
        \item In fact, by applying the reasoning in the above example we can simplify the RHS for any $f \geq 3$. Indeed, write $|\chi|^{2f} = 3^{r} (2-\tau)^s,$ where $r = \floor{\frac{f}{3}}$ and $s = f \pmod{3},$ the remainder of $f$ upon division by $3,$ therefore $0 \leq s \leq 2.$ In other words, we may always assume $(A,B,C) = 3^{r}(1,0,0)$ (when $s = 0$) or $3^{r}(2,-1,0)$ (when $s = 1$) or $3^r (4,-4,1)$ (when $s=2$).
    \end{enumerate}
\end{example}

We record our observations so far as follows.

\begin{lemma}
    Assume $f\geq 3$ and let $w_{1}(\xi) = \sum_{j=0}^{5} a_{j} \xi^j,$ $w_{2}(\xi) = \sum_{j=0}^5 b_j \xi^{j}$ and $w_3 = \sum_{j=0}^{5} c_{j} \xi^j$. If $z = \frac{1}{\chi^f} (w_1,w_2,w_3)$ is a unit vector then  $a,b,c \in \mathbb{Z}^{6}$ satisfy
    \begin{align*}
    q_{0}(a) + q_{0}(b) + q_{0}(c) = A\\
    q_{1}(a) + q_{1}(b) + q_{1}(c) = B\\
    q_{2}(a) + q_{2}(b) + q_{2}(c) = C
\end{align*}
for some integers $A,B$ and $C$ which are determined by $A + B\tau + C\tau^2 = 3^{r}(2-\tau)^{s},$ where $r = \floor{\frac{f}{3}}$ and $s =  f \pmod{3}.$ By the linear independence of $1, \tau, \tau^2$ over $\mathbb{Z}$, we can explicitly write down $A,B$ and $C$ for a given $f.$
\end{lemma}

\begin{theorem}
    For $p=3$, there exists finitely many (including zero) unit vectors with fixed sde $f$. The problem of listing all unit vectors with a given sde $f$ reduces to finding all the integer solutions of a positive-definite quadratic form in $18$ variables.
    
\end{theorem}
\begin{proof}  
First, observe that 
    \begin{align}
    q = q_{0} + 2q_{2}  = \alpha_0 - \alpha_3 
    &= \sum_{i=0}^{5} a_{i}^2 - a_{0}a_{3} - a_{1}a_{4} - a_{2}a_{5} \\
    & = (a_{0}^2 + a_{3}^2 - a_{0}a_{3}) + (a_{1}^2 + a_{4}^2 - a_{1}a_{4}) + (a_{2}^2 + a_{5}^2 - a_{2}a_{5}) \\
    & = \frac{1}{4} ((2a_{0} - a_{3})^2 + 3a_{3}^2 + (2a_{1} - a_{4})^2 + 3a_{4}^2 + (2a_{2} - a_{5})^2 + 3a_{5}^2)
\end{align}

\noindent which is a positive definite quadratic form.

\noindent By the above lemma, for a unit vector $z = \frac{1}{\chi^f} (w_1,w_2,w_3)$ of sde $f$, the three $6$-tuples $a,b,c \in \mathbb{Z}^6$ must satisfy $$q(a) + q(b) + q(c) = A + 2C$$  
\end{proof}

\begin{remark} \label{unitabs1}
    Suppose $a \in \mathbb{Z}^6$ such that $q(a) = 1$, then exactly one of the three integers $a_{0}^2 + a_{3}^2 - a_{0}a_{3}, a_{1}^2 + a_{4}^2 - a_{1}a_{4}$ or $a_{2}^2 + a_{5}^2 - a_{2}a_{5}$  is $= 1$ and the other two are $=0$. Assuming $a_{0}^2 + a_{3}^2 - a_{0}a_{3} = 1,$ we have $a = (\pm 1,0,0,0,0,0)$ or $ (0,0,0,\pm 1,0,0)$ or $(1, 0, 0, 1,0,0)$ or $(-1, 0, 0, -1,0,0).$ In other words, the corresponding element in $\mathbb{Z}[\xi]$ is $w_{1}(\xi) = \pm \xi^{3}$ or $1 + \xi^3$ or $-1 - \xi^3.$ Since $\xi$ satisfies its minimal polynomial $x^6 + x^3 + 1,$ we have $\xi^{6} = -(1+\xi^3).$ Therefore, $w_{1}(\xi) = \pm 1$ or $\pm \xi^3$ or $\pm \xi^6.$ Similarly when $a_{1}^2 + a_{4}^2 - a_{1}a_{4} = 1$ (resp. when $a_{2}^2 + a_{5}^2 - a_{2}a_{5} = 1$)  we have $w_{1}(\xi) = \pm \xi$ or $\pm \xi^4$ or $\pm \xi^7.$ (resp. $\pm \xi^2$ or $\pm \xi^5$ or $\pm \xi^8$). Therefore, all the possibilities for $w_{1}(\xi),$ namely $\pm \xi^{a},$ where $a \in \{0, \ldots, 8\}$ are in fact possible.  
\end{remark}

\begin{corollary}
    Let  $z = \frac{1}{\chi^f} (w_1,w_2,w_3) \in R_{\chi}^3$ be a unit vector such that $\chi \nmid w_{i}$ for all $i$, i.e, $sde(z,\chi) = f$.
    \begin{itemize}
        \item[(1.)] If $f = 0$ then $z = (\pm \xi^{a},0,0)$ up to permutation of coordinates, for some $a \in \mathbb{Z}.$ 
        \item[(2.)] If $f = 3$ then $z = \frac{1}{\chi^3}(\pm \xi^{a}, \pm \xi^{b}, \pm \xi^{c}).$
    \end{itemize}
\end{corollary}
\begin{proof}
    When $f=0,$ we have $q(a) + q(b) + q(c) = 1 + 2(0) = 1.$ Since $q$ is positive definite and takes integral values, $q(a) = 1$ and $q(b) = q(c) = 0$ up to permutation of $a,b$ and $c$. Then by Remark \ref{unitabs1}, the corollary follows. 

\noindent When $f=3,$ we have $q(a) + q(b) + q(c) = 3 + 2(0) = 3.$ Therefore $q(a) = q(b) = q(c) = 1.$ It follows that $(w_1(\xi),w_2(\xi),w_{3}(\xi)) = (\pm \xi^{a}, \pm \xi^{b}, \pm \xi^{c})$ for some integers $a,b$ and $c.$
\end{proof}

Now we argue that the unitary matrices in $U(3,R_{\chi})$ of sde $3$ are precisely the ones of the form $D_1 H D_2,$ for some matrices $D_1,D_2 \in U(3,R_{\chi})$ of sde $0.$ For this we first derive relations between two orthogonal unit vectors of sde $3.$

\begin{remark}
    Two unit vectors $z_{1} = \frac{1}{\chi^3}(\xi^{a_1},\xi^{a_2},\xi^{a_3})$ and $z_{2} = \frac{1}{\chi^3}(\xi^{b_1},\xi^{b_2},\xi^{b_3})$ of sde $3,$ are orthogonal iff $\xi^{a_1 - b_1} + \xi^{a_2 - b_2} + \xi^{a_3 - b_3} = 0$ iff $1 + \xi^p + \xi^q = 0,$ where $p = (a_2 - b_2) - (a_1 - b_1) = (a_2 - a_1) - (b_2 - b_1)$ and $q = (a_3 - b_3) - (a_1 - b_1) = (a_3 - a_1) - (b_3 - b_1).$ Note that by taking conjugate of this equation we get $\xi^p + \xi^q = \xi^{-p} + \xi^{-q}$ which implies $\xi^{p+q} = 1$ and therefore, $p + q = 0 \pmod{9},$ so $(p,q) = (3,6) \text{ or } (6,3) \pmod{9}.$ 
\end{remark}

\noindent From the above remark if ${z_{1},z_{2},z_{3}}$ are mutually orthogonal unit vectors of sde $3$, then after an appropriate permutation we may write the matrix with $z_i$ as columns as  $$[z_1, z_2, z_3] = \text{diag}(\pm \xi^{\alpha_1}, \pm \xi^{\alpha_2}, \pm \xi^{\alpha_3})* H* \text{diag}(\pm \xi^{\beta_1}, \pm \xi^{\beta_2}, \pm \xi^{\beta_3}),$$
for some integers $\alpha_i,\beta_i.$

\noindent This observation shows that the choice of gates Clifford$+\mathcal{D}$ is not arbitrary, as the $H$ gate naturally arises as the unique gate of sde $3$ modulo multiplication by matrices of sde $0$ on both sides.  

\section{Main results}
\label{mainresults}
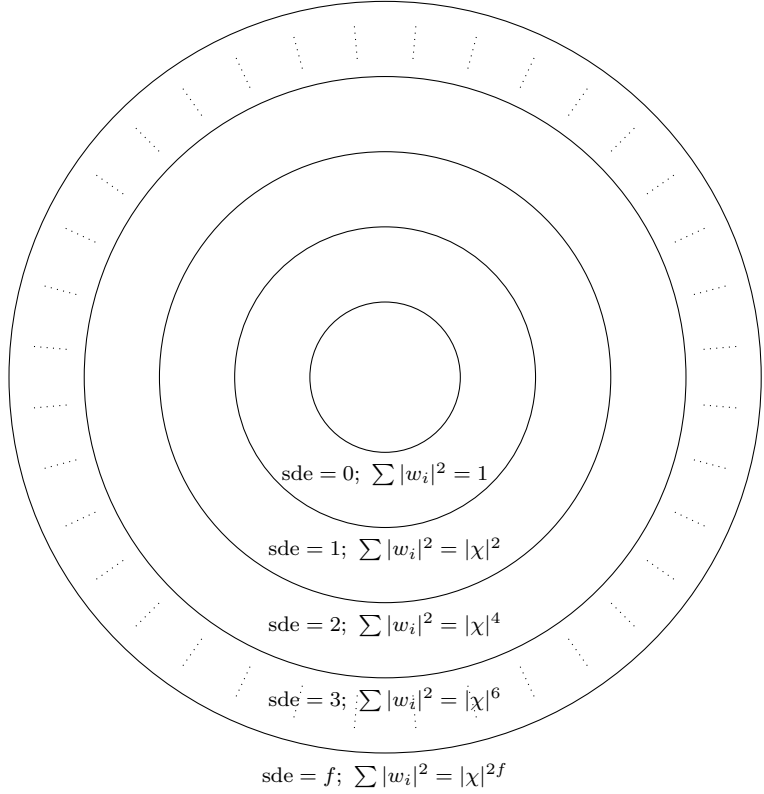
\begin{figure}[h]
    \centering
    \begin{tikzpicture}
    
    \foreach \r in {1,2,3,4,5}{
        \draw (0,0) circle (\r);
    }
    
    \foreach \r/\label in {
        1/{\scriptsize $\text{sde}=0;~\sum|w_{i}|^{2}=1$},
        2/{\scriptsize $\text{sde}=1;~\sum|w_{i}|^{2}=|\chi|^{2}$},
        3/{\scriptsize $\text{sde}=2;~\sum|w_{i}|^{2}=|\chi|^{4}$},
        4/{\scriptsize $\text{sde}=3;~\sum|w_{i}|^{2}=|\chi|^{6}$},
        5/{\scriptsize $\text{sde}=f;~\sum|w_{i}|^{2}=|\chi|^{2f}$}
    }{
        \node[anchor=north] at (0,-\r) {\label};
    }
    
    \foreach \angle in {5,15,...,355}{
        \ifnum\angle<360
            \draw[dotted] (\angle:4.25) -- (\angle:4.75) {};
        \fi
    }
    \end{tikzpicture}
    \caption{Unit vectors $R_{n,\chi}^{p}$ can be partitioned by their smallest denominator exponent. By multiplying a unit vector of a given sde by a unitary $U$ in $U(p,R_{n,\chi})$ we may change its sde. Having control of this change by choosing an appropriate unitary $U$ helps us in exact synthesis in $U(p,R_{n,\chi})$. In what follows we study the possibility of this phenomenon when $n = p^{l} = 3,8,9$. }
\end{figure}

In the following lemma we observe that the sde of the entries of a given unit vector  in $R_{8,\chi}^2$ or $R_{3,\chi}^3$ or $R_{9,\chi}^3$ are equal.   
\begin{lemma}\label{equalsde}~
\begin{enumerate}[(i)]
\item Suppose $z = (z_1,z_2)^{T}$ is a unit vector whose entries are in $R_{8,\chi}$. Then $\text{sde}(z_1) = \text{sde}(z_2)$
\item Suppose $z = (z_1,z_2,z_3)$ is a unit vector whose entries are in $R_{3,\chi}$ or $R_{9,\chi}$. Then $\text{sde}(z_1) = \text{sde}(z_2) = \text{sde}(z_3)$
\end{enumerate}
\end{lemma}
\begin{proof}
Write $z = \frac{1}{\chi^f} (w_1,w_2)^T$ in case 1 and $z = \frac{1}{\chi^f} (w_1,w_2, w_3)^T$ in case 2. By writing the unit vector condition for $z$ and applying the maps $P : \mathbb{Z}[\omega] \xrightarrow{} \mathbb{Z}_2$ and $P : \mathbb{Z}[\xi] \xrightarrow{} \mathbb{Z}_3$ respectively, we get $w_1(1)^2+ w_{2}(1)^2 = 0  \pmod{2}$ and $w_{1}(1)^2 + w_{2}(1)^2 + w_{3}(1)^2 = 0 \pmod{3}$ respectively.

\noindent The solutions of $w_1(1)^2+ w_{2}(1)^2 = 0  \pmod{2}$ or $w_{1}(1)^2 + w_{2}(1)^2 + w_{3}(1)^2 = 0 \pmod{3}$  are such that $w_{i}(1) = 0$ for all $i$ or $w_{i}(1) \neq 0$ for all $i$.   
\end{proof}

\begin{remark}
    Moreover, we can argue that given a unitary matrix $U \in U(p,R_{n,\chi}),$ when $n = p^l$ and $p \leq 3$ the sde of all the entries of $U$ are equal. Indeed, note that each entry $u_{ij}$ of $U$ is shared by the $i$-th row $u_i$ and the $j$-th column $w_j$ of $U$. Since both $u_i$ and $w_j$ are unit vectors in $R_{n,\chi},$ so the claim follows by the above lemma. Therefore, the notions of sde of an element, a vector and a matrix coincide when $p = 2$ or $3.$  
\end{remark}

\subsection{Qubit Clifford$+T$}
\label{qubit}
In this section, we give an alternative proof of the main ideas in \cite{kmm} using the `derivatives mod p' trick introduced earlier. This also gives us an alternative approach for exact synthesis of a qubit gate $U$ with entries in $\mathbb{Z}[i,\frac{1}{\sqrt{2}}]$. In particular, given a unitary $U$ in $U(2,R_{8,\chi})$ of sde $\geq 3$, we consider its first column $z$ and find an integer $k$ such that $\text{sde}(HT^{k}z) = \text{sde}(z) - 1.$ We do this by computing the derivative of the entry of $HT^{k}z$ after clearing the denominator $\sqrt{2}$ and by solving certain linear equations mod $2$. 
 This is precisely part of lemma 3 of \cite{kmm}. By proceeding inductively we find integers $k_1, \ldots, k_r$ such that $$ \text{sde}(HT^{k_1}\ldots HT^{k_r}u) = 2$$ 

\noindent We already know all the unit vectors of $\text{sde} = 2$ or lower (in fact, we can find them by the qubit analogue of section \ref{unit vectors}). Therefore, we recover $U$ as a word in $H,T$ and a $\text{sde} = 2$ matrix. The following remark highlights the simple yet important relation between $\chi$ and $\sqrt{2}$.

\begin{remark}
    In $\mathbb{Z}[\zeta_8],$ we have the identity
    $$2 = (1-\zeta_8)(1-\zeta_8^3)(1-\zeta_8^5)(1-\zeta_8^7) = u\chi^4$$
    where $\chi = 1 - \zeta_8$ and $u = \frac{(1-\zeta_8^3)}{(1-\zeta_8)}\frac{(1-\zeta_8^5)}{(1-\zeta_8)}\frac{(1-\zeta_8^7)}{(1-\zeta_8)}.$ This can be obtained by the fact that $\Phi_8(x) = x^4 + 1 = (x-\zeta_8)(x-\zeta_8^3)(x-\zeta_8^5)(x-\zeta_8^7).$ Here $u$ is a unit in $\mathbb{Z}[\zeta_8].$ We can further simplify this identity by observing that $\zeta_8^5 = \zeta_8^{-3}$ and $\zeta_8^7 = \zeta_8^{-1}$ which gives $2 = \zeta_8^4(1-\zeta_8)^2(1-\zeta_8^3)^2$, i.e., $2$ is a perfect square $\mathbb{Z}[\zeta_8]$ which further implies that $\sqrt{2} = \zeta_8^2 (1-\zeta_8)(1-\zeta_8^3) \in \mathbb{Z}[\zeta_8].$ This observation immediately gives us that $sde(z,\chi) = 2sde(z,\sqrt{2}).$ 
\end{remark}

We now the prove the lemmas stated in \cite{kmm} using our methods. Note that we write $\text{sde}(z,\chi)$ as $\text{sde}(z)$.

\begin{lemma} (Lemma 2 in \cite{kmm})
Let 
$\begin{pmatrix}
z_{1}\\
z_{2}\\
\end{pmatrix}$
be a unit vector with entries in $\bZ[\zeta_{8}]_{\chi}$ and let $\text{sde}(z,\chi)\geq 4$, Then for any integer $k$:
\begin{equation}
-1 \leq \text{sde}\left(\frac{z_{1}+\zeta_{8}^{k}z_{2}}{\sqrt{2}}\right)-\text{sde}(z_{1}) \leq 1
\label{eqnsdelemma2}
\end{equation}
\end{lemma}
\begin{proof}
Now recall that 
\[
HT^{k}z=
\begin{pmatrix}
\frac{z_{1}+\zeta_{8}^{k}z_{2}}{\sqrt{2}}\\
\frac{z_{1}-\zeta_{8}^{k}z_{2}}{\sqrt{2}}\\
\end{pmatrix}
=\frac{1}{\sqrt{2}\chi^{f}}
\begin{pmatrix}
{q_{1}+\zeta_{8}^{k}q_{2}}\\
q_{1}-\zeta_{8}^{k}q_{2}\\
\end{pmatrix}
\]
Let $q:= q_{1}+\zeta_{8}^{k}q_{2}$ and $q' = q_{1}-\zeta_8^{k}q_{2}.$  Note that $\text{gde}(q) = \text{gde}(q')$ by Lemma \ref{equalsde}. Let $l$ be the largest integer such that $\chi^{l}$ divides $q = q_{1}+\zeta_{8}^{k}q_{2}$, i.e., $l = \text{gde}(q)$. Suppose $l \geq 5,$ i.e., $\chi^{5}$ divides $q$ and $q'$. Since $q+q' = 2q_1,$ we have $\chi^5$ divides $2q_1$. Since $2 = \chi^{4}u$ for some unit $u \in \mathbb{Z}[\zeta_{8}],$ we may conclude that $\chi$ divides $q_1$ which is a contradiction. Therefore, $l \leq 4.$

\medskip

\noindent On the other hand, $q(1) = q_{1}(1) + q_{2}(1) = 0 \pmod{2}.$ Therefore, $\chi$ divides $q(1)$ and hence $l \geq 1.$

Since $\chi^2$ divides $\sqrt{2},$ we have $$\text{sde}(HT^{k}z) = \text{sde}(z)+2-l$$
Therefore,
\[
\text{sde}(HT^{k}{z})-\text{sde}({z}) = 2-l
\]

Since $1 \leq l \leq 4$
\[
-1 \leq \text{sde}\left(\frac{z_{1}+\zeta_{8}^{k}z_{2}}{\sqrt{2}}\right)-\text{sde}(z_{1}) \leq 2
\]

\noindent If $\text{sde}(z) \geq 3$ then we may in fact conclude that 
$\text{sde}\left(\frac{z_{1}+\zeta_{8}^{k}z_{2}}{\sqrt{2}}\right)-\text{sde}(z_{1}) \leq 1$ by a finer analysis of $q$ and $q'$. However, we skip the proof.
\end{proof}

\noindent Let $z$ = $\begin{pmatrix}
z_{1}\\
z_{2}\\
\end{pmatrix} = \frac{1}{\chi^f} \begin{pmatrix}
q_{1}\\
q_{2}\\
\end{pmatrix}$ and $q(\zeta_{8}) = q_{0} + \zeta_{8}^k q_{1}$ as above. 
 Note that the first three derivatives of $q(\zeta_{8})$ (starting from $0$-th) at $\omega = 1$ are $q(1) = q_0(1) + q_1(1),$ $q'(1) = q_0'(1) + q_1'(1) + k q_1(1)$ and $\frac{q^{(2)}(1)}{2!} = \frac{q_{1}^{(2)}(1)}{2!} + \frac{q_{2}^{(2)}(1)}{2!} + kq_{2}'(1) + {k\choose 2}q_2(1).$ If $k = 2t + e,$ where $e = 0 \text{ or } 1$ then we note that ${k \choose 1} = e \pmod{2}$ and ${k \choose 2} = t \pmod{2}.$ This set-up immediately yields us the following theorem. 

\begin{theorem} (Lemma 3 in \cite{kmm})
    Let $z = \begin{pmatrix}
z_{1}\\
z_{2}\\
\end{pmatrix} = \frac{1}{\chi^f} \begin{pmatrix}
q_{1}\\
q_{2}\\
\end{pmatrix}$ be a unit vector in $R_{\chi}^2$ such that $q_1,q_2 \in R$ and $\chi \nmid q_1,q_2$ then for each $s \in \{-1,0,1\}$ there always exists $k \in \{0,1,2,3\}$ such that $$\text{sde}(HT^{k}z) - \text{sde}(z) = s$$ 
\end{theorem}
\begin{proof}
    First, we notice that $q(1) = q_0(1) + q_1(1) = 0 \pmod{2}$ by default. One can choose $e = -q_1'(1) - q_2'(1) $ so that $q'(1) = 0$ and choose $t = -(\frac{q_{1}^{(2)}(1)}{2!} + \frac{q_{2}^{(2)}(1)}{2!} + e q_{2}'(1))$ so that $\frac{q^{(2)}(1)}{2!} = 0.$ In summary, we can choose an integer $0 \leq k \leq 3$ such that $q(1) = q'(1) = \frac{q^{(2)}(1)}{2!} = 0 \pmod{2}.$ By Theorem \ref{derivativecriterion} for $n=8,$ $\chi^{3}$ divides $q(\zeta_{8}),$ i.e, $\text{sde}(HT^{k}z) - \text{sde}(z) = 2 - l = -1.$

\noindent In a similar vein, we can choose $k$ such that $f(1) = 0$ and $f'(1) \neq 0$ or $f(1) = f'(1) = 0$ and $f''(1) \neq 0.$
\end{proof}

\subsection{Qutrit Clifford$+R$}
\label{cliffordR}
 Let $U \in U(3,R_{\chi})$, where $R = \mathbb{Z}[\omega]$, where $\omega = e^\frac{2\pi i}{3},$ $\chi = 1 - \omega$ and $z = \frac{1}{\chi^{f}}(p(\omega),q(\omega),r(\omega))^{\text{T}}$ be the first column of $U$. Here $p,q$ and $r$ are polynomials with integer coefficients such that $\chi \nmid p(\omega), q(\omega)$ and $r(\omega)$. Here $f = sde(z,\chi)$ (denoted as $sde(z)$) by definition. Since $\omega$ has $x^2 + x + 1$ as the minimal polynomial over $\mathbb{Z}$, we may assume that $p,q,r$ are polynomials of degree at most $1$. In the following remark, we show that $3 = \chi^{2} u$ for some unit $u$ in $\mathbb{Z}[\omega].$ Therefore, $sde(H) = 1.$

\begin{remark}
    Recall that $\Phi_3(x) = x^2 + x + 1 = (x - \omega)(x-\omega^2),$ which gives us that $3 = (1-\omega)(1-\omega^2) = (1-\omega)^2 \frac{1-\omega^2}{1-\omega} = \chi^2 u,$ where $u = \frac{1-\omega^2}{1-\omega}$ is a unit in $\mathbb{Z}[\omega].$ In fact, $ \frac{1-\omega^2}{1-\omega} = 1+\omega = -\omega^2.$ Therefore, $-3 = \omega^2 (1-\omega)^2$ is a perfect square in $\mathbb{Z}[\omega],$ which implies $\sqrt{-3} = \omega (1-\omega) = \omega \chi.$
\end{remark}

\begin{theorem}
    For every unit vector $z \in R_{\chi}^{3}$ of sde $f \geq 1$, there exists integers $0 \leq a_0, a_1, a_2 \leq 2, \epsilon \in \{0,1\}$ and $\delta \in {0,1,2}$ such that $\text{sde}(HDR^{\epsilon}X^{\delta}z, \chi) < \text{sde}(z,\chi),$ where $R = \text{diag}(1,1,-1)$ and $D = \text{diag}(\omega^{a_0}, \omega^{a_1}, \omega^{a_2})$.
    \label{cliffrthm}
\end{theorem}

\begin{proof}
    We prove this theorem by providing an algorithm for choosing $\epsilon, \delta$ and $D$. 
    First note that $(p(1),q(1),r(1)) \pmod{3} = \pm (1,1,1) \text{ or } \pm (1,1,-1) \text{ or } \pm (1,-1,1) \text{ or } \pm (-1,1,1)$. The latter three cases can be reduced to $\pm (1,1,-1)$ by applying the Pauli gates $X$ or $X^2$. Indeed, the unit vector condition for $z$ can be written as $$|p(\omega)|^2 + |q(\omega)|^2 + |r(\omega)|^2 = |\chi|^{2f} $$
    By taking the image of the map $P : R \xrightarrow{} \mathbb{Z}_3$, we obtain the equation $p(1)^2 + q(1)^2 + r(1)^2 = 0$ in $\mathbb{Z}_3.$ Note that $p(1),q(1)$ and $r(1) \neq 0 \pmod{3}$ as $\chi$ does not divide any of them. Therefore, you get the possibilities listed above for $(p(1), q(1), r(1))$. 
    
\noindent Choose $\epsilon = 1$ if $(p(1),q(1),r(1)) = \pm(1,1,-1)$ and $\epsilon = 0$ otherwise. In other words, we may replace $z$ with $R^{\epsilon}X^{\delta}z$ for appropriate $\epsilon \in \{0,1\}, \delta \in \{0,1,2\}$ and assume that $(p(1),q(1),r(1)) = \pm (1,1,1).$

\medskip
    
\noindent Now, consider the polynomial $f(x) = x^{a_0}p(x) + x^{a_1}q(x) + x^{a_2}r(x).$ Note that ${f(\omega)}$  is precisely the first coordinate of $\chi^{f+1}HDR^{\epsilon}z$ up to multiplication by a unit in $\mathbb{Z}[\omega]$. Therefore, by Lemma \ref{equalsde}, $\text{sde}(HDU) = \text{sde}(HDz) = \text{sde}(\frac{f(\omega)}{\chi^{f+1}}) = f+1 - k,$ where $k$ is the largest integer for which $\chi^{k}$ divides $f(\omega).$ So $\text{sde}(HDz,\chi) - \text{sde}(z,\chi) = 1-k < 0$ iff $k \geq 2$, i.e., $\chi^{2}$ divides $f(\omega).$ 
    
\noindent Using the notion of derivatives mod $3$ and Theorem \ref{derivativecriterion}, we have that $\chi^2$ divides $f(\omega)$ iff $$f(1) = p(1) + q(1) + r(1) =  0 \pmod{3}$$ and $$f'(1) = a_{0}p(1) + a_1q(1) + a_2 r(1) + (p'(1) + q'(1) + r'(1)) = 0 \pmod{3}.$$ 

\noindent The first condition is automatically satisfied as $(p(1),q(1),r(1)) = \pm (1,1,1).$  Since each of $p(1), q(1)$ and $r(1)$ are non-zero, there exists plenty of triples $(a_0,a_1,a_2) \in \mathbb{Z}_{3}^{3}$ for which the second equation is satisfied - simply choose arbitrary $a_1$ and $a_2$ then determine $a_0$.
\end{proof}

\begin{corollary} \label{anyons}
    Clifford$+R$ = $U_{3}(\mathbb{Z}[\omega]_{\chi})$
\end{corollary}
\begin{proof}
    By the above theorem, for each $U \in U_{3}(\mathbb{Z}[\omega]_{\chi})$ of sde at least $1,$ there exists a syllable of the form $HDR^{\epsilon}X^{\delta}$ such that $\text{sde}(HDR^{\epsilon} X^{\delta}z) \leq \text{sde}(z)-1.$ Therefore, by proceeding inductively there exists a sequences of pairs $(D_1,\epsilon_1, \delta_1), \ldots (D_k, \epsilon_k, \delta_k)$ such that $$\text{sde}(HD_{1}R^{\epsilon}X^{\delta_1}HD_{2}R^{\epsilon_2} \ldots HD_{k}R^{\epsilon_k}X^{\delta_k} U) = 0. $$ 
    We know that the matrices of sde $0$ in $U_{3}(\mathbb{Z}[\omega]_{\chi})$ are precisely the matrices generated by \newline $\text{diag}(\pm \omega^{a_0}, \pm \omega^{a_1}, \pm \omega^{a_2})$, $H^2$ and the Pauli gates.
\end{proof}

\begin{algorithm}[H]
\caption{Algorithm for Clifford$+R$ exact synthesis.}\label{algocliffordr}
\begin{algorithmic}
    \State \textbf{Input:} $U \in U(3, R_{3,\chi})$, let $z = U\ket{0}$
    \While {$\text{sde}(z) > 0$}
    \State \textbf{choose}: $(D,\epsilon,\delta)$, where $D = \text{diag}( \omega^{a_0}, \omega^{a_1}, \omega^{a_2})$ and $\epsilon\in \{0,1\}$ and $\delta \in \{0,1,2\}$
    \Statex \hspace{\algorithmicindent} 
    \textbf{such that} $\text{sde}(HDR^{\epsilon} X^{\delta}z) \leq \text{sde}(z) - 1$
    \State \textbf{update} $z = HDR^{\epsilon} X^{\delta} z$
    \Statex \EndWhile
\end{algorithmic}
\end{algorithm}
\begin{remark}
\noindent The while loop in Algorithm \ref{algocliffordr} is guaranteed to terminate since $\text{sde}(z)$ decreases in each iteration as shown in Theorem \ref{cliffrthm}.
\end{remark}
\subsection{Qutrit Clifford$+\mathcal{D}$}

Now we turn to the case $n = 3^2,$ i.e., $p=3$ and $l=2$. We denote $\xi = e^{\frac{2\pi i}{9}}$, a primitive $9$-th root of unity. Unlike for $n=3$ and $8$ dealt with in the previous sections, for $n=9$ the property $P$ fails to hold for a certain set of unit vectors. To be more precise, for each unit vector $z \in R_{9,\chi}^3$ there exists a $\Delta \in \mathbb{Z}_3$ such that $\text{sde}(HDR^{\epsilon}X^{f}z) - \text{sde}(z) = -1$ if and only if $\Delta \neq -1.$ Moreover, the $\Delta$ can be easily described in terms of the derivatives mod $3$ of the entries of $z.$ 

\medskip

\noindent We begin by deriving certain basic properties of derivatives $\pmod{3}$ of $w_{i}$ imposed by the unit vector condition, where $z = \frac{1}{\chi^{f}}(w_1, w_2, w_3)$ and $\chi \nmid w_{i}$ for some $i$. Recall that by applying the map $P : w_{i}(\xi) \mapsto w_{i}(1) \pmod{3}$ to the unit vector condition we may in fact conclude that $\chi \nmid w_{i}$ for all $i$. Moreover, for each $w_{i}$ there is a unique vector $(w_{i}(1), \ldots, \frac{w_{i}^{(5)}(1)}{5!}) \in \mathbb{Z}_{3}^{6}$ associated to it as discussed Section \ref{derivative}. Recall that this vector is simply the image of $w_{i}$ via the quotient map $\mathbb{Z}[\xi] \xrightarrow{} \mathbb{Z}[\xi]/(3)$ written with respect to the basis $\{1, \Bar{\chi}, \ldots, \Bar{\chi}^5\}.$ Here $\mathbb{Z}[\xi]/(3)$ is a $6$ dimensional vector space over $\mathbb{Z}_3$ with $\{1, \Bar{\chi}, \ldots, \Bar{\chi}^5\}$ as a basis.   

\medskip

\noindent If the vector $z$ has sde $\geq 1$ the unit vector condition imposes more restrictions on $\frac{w_{i}^{(k)}(1)}{k!}$ also for $k \geq 1$. 
We can similarly apply the higher analogues of the map $P,$ defined in  section \ref{derivative}, namely the maps $P^{(k)} : R \xrightarrow{} \mathbb{Z}_3,$ defined by $P^{(k)}(f(\zeta)) = \frac{f^{(k)}(1)}{k!}$ to get the restrictions on higher derivatives of $w_i.$ Note that $|w_{i}(\zeta)|^2 = w_{i}(\zeta)\overline{w_{i}(\zeta)} = w_{i}(\zeta)w_{i}(\zeta^{-1})$. The derivative of $w_{i}(x)w_{i}(x^{-1})$ is $ w_{i}'(x)w_{i}(x^{-1}) + w_{i}(x) w_{i}'(x^{-1})\frac{-1}{x^2},$ which is $0$ at $x=1$. Therefore, we get a vacuous equation `$0 = 0$' by applying $P^{(1)}$ on both sides of Equation \ref{unitvector}.

\noindent From now on, let us follow the notation $(p,q,r) := (w_1,w_2,w_3)$ whenever we work with qutrits to avoid cumbersome double indices. We further denote, $p_k := \frac{p^{(k)}(1)}{k!}$ (resp. for $q_k$ and $r_k$) for $0 \leq k \leq 5.$ With this notation we get the following equation by applying $P^{(2)}$ on both sides of Equation \ref{unitvector}$$2p_{2}p_{0} + p_{1}^2 - \frac{p_0 p_1}{2} + 2q_{2}q_{0} + q_{1}^2 - \frac{q_0 q_1}{2} + 2r_{2}r_{0} + r_{1}^2 - \frac{r_0 r_1}{2} = 0.$$ 

\noindent We record this expression for future use in the case where $(p_0,q_0,r_0) = (1,1,1).$ Note that all the $p_k,q_k,r_k$'s are elements of $\mathbb{Z}_3.$ 

\begin{lemma} \label{secondder}
Let $\pi_i := p_i + q_i + r_i$ for $i=1,2$. If $(p_0,q_0,r_0) = (1,1,1)$ then $$-\pi_2 + \pi_1 +  p_1^2 + q_1^2 + r_1^2  = 0$$
\end{lemma}

\noindent For the qutrit Hadamard gate $H,$ recall that $\text{sde}(H,\chi) = 3$ as the denominator of $H$ is $\sqrt{-3} = u \chi^{3}$ for some unit $u \in \mathbb{Z}[\xi],$ where $\chi = 1 - \xi.$  For a unitary matrix $U \in U(3,R_{\chi}),$ by lemma \ref{equalsde}, the sde's of every coordinate of $Uz$ are equal.  
  
\noindent Therefore, $\text{sde}(HDu) \leq \text{sde}(u) - 1$ for some $D = \text{diag}(\xi^{a_1},\xi^{a_2},\xi^{a_3})$ if and only if $\chi^4$ divides $f(\xi) = \xi^{a_1}p(\xi) + \xi^{a_2}q(\xi) + \xi^{a_3}r(\xi),$ the first coordinate of $\chi^{3+f}HDz$ up to a unit in $\mathbb{Z}[\xi]$. Theorem \ref{maintheorem} essentially gives a criterion for the existence of $a_{1},a_{2}$ and $a_3$ such that $\chi^4$ divides $f(\xi).$ For this to work we may have to permute the coordinates of $z$ using the $X$ gate and multiply a specific coordinate by $-1$ using the $R$ gates. We begin with a proposition which reduces the problem to that of existence of $a_{1},a_{2}$ and $a_3$ such that $\chi^3$ divides $f(\xi).$

\begin{proposition} \label{reductionprop}
    Let $z \in R_{\chi}^3$ be a unit vector. If there exists a $D = \text{diag}(\xi^a,\xi^b,\xi^c)$ such that $\text{sde}(HDz,\chi) \leq \text{sde}(z,\chi)$ then there exits a $D' = \text{diag}(\xi^{a'},\xi^{b'},\xi^{c'})$ such that $\text{sde}(HD'z,\chi) < \text{sde}(z,\chi).$
\end{proposition}
\begin{proof}
    Let $(a,b,c) = (3k_1+ \epsilon_1, 3k_2+ \epsilon_2, 3k_3+ \epsilon_3).$ Since $\text{sde}(H) = 3$ the hypothesis is satisfied if and only if $f(1) = \frac{f^{(1)}(1)}{1!} = \frac{f^{(2)}(1)}{2!} = 0 \pmod{3}$. But the first three derivatives depend only on $\epsilon_i.$
    
\noindent Now, we may choose $(a',b',c') = (3k_1'+ \epsilon_1, 3k_2'+ \epsilon_2, 3k_3'+ \epsilon_3)$  so that  
\begin{align}
    \frac{f^{(3)}(1)}{3!} &= k_1' p_0 + k_2' q_0 + k_3' r_0 + (\epsilon_1 - \epsilon_1^{2})p_1 + (\epsilon_2 - \epsilon_2^{2})q_1 + (\epsilon_3 - \epsilon_3^{2})r_1 \\
    &\quad + \epsilon_1 p_2 + \epsilon_2 q_2 + \epsilon_3 r_3  + (p_3 + q_3 + r_3) = 0
\end{align}
    is satisfied. We can always do this as $p_0,q_0$ and $r_0$ are non-zero.
\end{proof}

\begin{theorem} \label{maintheorem}
    Let $z \in R_{\chi}^{3}$ be a unit vector, where $R = \mathbb{Z}[\xi]$ and $\chi = 1 - \xi$. There exists a polynomial $g(x_{1},x_{2}) = (x_1 - x_2)^{2} + \alpha(x_1 - x_2) + \beta \in \mathbb{Z}_{3}[x_{1},x_{2}]$ such that $\text{sde}(HDR^{\epsilon}z) \leq \text{sde}(z) - 1$ iff $g(x_1,x_2) = 0$ has a solution in $\mathbb{Z}_{3}^2$ iff $\Delta = \alpha^{2} - \beta \neq - 1.$
\end{theorem}

\begin{proof}
As it was discussed in Section 5.2, for the unit vector $z = \frac{1}{\chi^f}(p,q,r)^{\text{T}},$ either $(p(1),q(1),r(1)) = \pm (1,1,1)$ or $\pm (1,1,-1) \text{ or } \pm (1,-1,1) \text{ or } \pm (-1,1,1).$ The latter two cases can be reduced to $\pm (1,1,-1)$ by applying the Pauli gates $X$ or $X^2$. Now, for the first case choose $\epsilon = 0$ and for the latter, $\epsilon = 1$ i.e., $\epsilon$ is chosen in such a way that $R^{\epsilon}(p(1),q(1),r(1)) = \pm (1,1,1)$. In other words, we may replace $z$ with $R^{\epsilon}X^{f}z$ for appropriate $\epsilon \in \{0,1\}, f\in \{0,1,2\}$ and assume that $(p(1),q(1),r(1)) = \pm (1,1,1).$

\noindent To choose $D$ such that $\text{sde}(HDz) \leq \text{sde}(z) - 1$ is equivalent to choosing $0 \leq a_{1}, a_{2}, a_{3} \leq 8$ such that $\chi^{4}$ divides $f(\xi).$ By Theorem \ref{derivativecriterion}, $\chi^4$ dividing $f(\xi)$ is equivalent to $f(1) = \frac{f^{(1)}(1)}{1!} = \frac{f^{(2)}(1)}{2!} = \frac{f^{(3)}(1)}{3!} = 0 \pmod{3}.$ Note that $f(1) = p(1) + q(1) + r(1) \pmod{3} = 0$ is automatically satisfied. Furthermore, by Proposition \ref{reductionprop}, if there exists a diagonal cyclotomic $D$ such that $\text{sde}(HDu) \leq \text{sde}(z),$ there exists a $D'$ such that $\text{sde}(HD'z) \leq \text{sde}(z) - 1$. Therefore, we are reduced to finding $a_{i}$'s such that $\frac{f^{(1)}(1)}{1!} = \frac{f^{(2)}(1)}{2!} = 0 \pmod 3$, i.e., $$a_{1}p(1) + a_{2}q(1) + a_{3}r(1) + p'(1) + q'(1) + r'(1) = 0$$ and 
$${a_{1}\choose2}p(1) + {a_2 \choose 2} q(1) + {a_3 \choose 2 }r(1) + a_{1}p'(1) + a_{2}q'(1) + a_{3}r'(1) + \frac{p^{(2)}(1)}{2!} +  \frac{q^{(2)}(1)}{2!} + \frac{r^{(2)}(1)}{2!} = 0$$ 

They simplify as $$\pi_1 + \sum_{i=1}^{3} a_{i}  = 0$$ and $$\sum_{i=1}^{3} (a_i - a_{i}^{2}) + a_{1}p'(1) + a_{2}q'(1) + a_{3}r'(1) + \pi_2 = 0$$
respectively, where $\pi_{i} = p_{i} + q_{i} + r_{i}$.
By writing $a_3 = -\pi_1 - a_1 - a_2$ and substituting in the second equation, we get a polynomial $g(a_{1},a_{2}) = -\pi_1 + \pi_{2} -a_1^2 - a_2^2 - (-\pi_1 - a_1 - a_2)^2 + a_{1}p'(1) + a_{2}q'(1) + (-\pi_1 - a_{1} -a_{2})r'(1) = -\pi_1 + \pi_{2} -\pi_{1}^2  -\pi_1 r'(1) + (\pi_1 + p'(1) - r'(1))a_{1} + (\pi_1 + q'(1) -r'(1))a_2 + (a_{1}^2 + a_2^2 + a_{1}a_{2}) = \gamma + \alpha (a_{1} - a_{2}) + (a_1 - a_2)^2,$ where $$\alpha = q'(1) - p'(1)$$ and $$\gamma = -\pi_1 + \pi_{2} -\pi_{1}^2  -\pi_1 r'(1) $$

\noindent Furthermore, $g(a_1,a_{2}) = 0$ has a solution in $\mathbb{Z}_{3}^2$ iff $ \Delta = \alpha^{2} - \gamma \neq -1$

\end{proof}
\begin{remark}
    
\noindent We explicitly describe the element $\Delta := \alpha^2 - \gamma \in \mathbb{Z}_3$ associated to the unit vector $z$ in terms of the derivatives of the numerators of entries of $z$ in their reduced form in $R_9,\chi.$ From the formulae for $\Delta$ in the proof, we have $\Delta = \alpha^2 - \gamma = (q_1 - p_1)^2 - (-\pi_1 + \pi_2 - \pi_1^2 - \pi_1r_1) = q_1^2 + p_1^2 + p_1q_1 + \pi_1 - \pi_2 + p_1^2 + q_1^2 + r_1^2 + 2p_1q_1 + 2q_1r_1 + 2p_1r_1 + p_1r_1 + q_1r_1 + r_1^2 = 2(p_1^2 + q_1^2 + r_1^2) + \pi_1 - \pi_2.$

\noindent But if $(p_0,q_0,r_0) = (1,1,1)$ then by Lemma \ref{secondder} we have,  $$-\pi_2 + \pi_1 +  p_1^2 + q_1^2 + r_1^2  = 0.$$Therefore, $\Delta = p_1^2 + q_1^2 + r_1^2.$

\end{remark}

\section{Conclusion}
In this paper we provide a general framework for qutrit synthesis over several families of universal gate sets. We used the notion of `derivatives mod $p$' to reproduce and extend the results in \cite{kmm} and \cite{vadymanyons}. 
This is achieved by relating the problem of synthesis to solving a system of polynomial equations mod $p$. We explored this for $p=2$ and $3$ in this paper.
On the other hand, we can try applying the same technique to multiqutrit synthesis. 

\noindent While the above method helped us derive the results in Section $5$, it remains unknown whether $U(3,R_{9,\chi})$ is finitely generated where the maximum sde of a generating set is relatively small. In such a case, we can in fact search for such a generating set using the methods of Section $4$. Consequently, this could help us achieve the exact analogue of \cite{kmm} for qutrits.

\noindent Moreover, the techniques in this paper do not fully apply for $p \geq 5$, as Lemma \ref{equalsde} fails. Indeed, for $p = 5$ we see that $2^2 + 1^2 + 0^2 + 0^2 + 0^2 \equiv 0 \pmod{5}.$ This suggests that there could be unit vectors such that the sde of the entries are not equal to each other. Therefore, the analogous systems of polynomials over $\mathbb{Z}_p$ are far too many to analyze as in the previous section for $p=2,3.$ We intend to explore this direction in the future. 

\label{conclusion}
\section{Acknowledgements}
ARK and DV would like to thank Neil Julien Ross, Peter Selinger, Jon Yard, Shiroman Prakash, Yash Totani, Vadym Kliuchnikov, Richard Cleve and David Gosset for sharing their understanding of various aspects of circuit synthesis and for suggesting relevant references at different times of the project. A majority of this work was done while DV was at IQC as a postdoc, he thanks David Jao for support. 

\noindent ARK and MM thank NTT research for financial and technical support. Research at IQC is supported in part by the Government of Canada through Innovation, Science and Economic Development Canada (ISED). 

\bibliography{qutrit}
\bibliographystyle{plainnat}
\end{document}